\documentclass[letterpaper,11pt]{article}
\usepackage{dsfont}
\usepackage{float}
\usepackage{tgpagella}
\usepackage[utf8]{inputenc}
\usepackage{graphicx,fullpage,paralist}
\usepackage{amsmath, amssymb, amsthm}
\usepackage{comment,hyperref}
\usepackage{thmtools}
\usepackage{tikz}
\usepackage{pgfplots}
\usepackage{varwidth}
\pgfplotsset{compat=1.10}
\usepgfplotslibrary{fillbetween}
\usetikzlibrary{backgrounds}
\usetikzlibrary{patterns}

\newcommand{\T}{\mathcal{T}}

\newcommand{\NN}{\mathbb{N}}

\newcommand{\OPT}{\mathrm{opt}}
\global\long\def\xc{\mathrm{xc}}%
\global\long\def\SAT{\mathrm{SAT}}%
\global\long\def\C{\mathcal{C}}%
 \global\long\def\Jd{J_{\mathrm{dum}}}%
\global\long\def\Md{M_{\mathrm{dum}}}%
\newcommand{\cmulti}{\mathcal{C}_{*}}

\newtheorem{theorem}{Theorem}
\newtheorem{lemma}{Lemma}
\newtheorem{corollary}{Corollary}

\newtheorem{proposition}{Proposition}

\usepackage[textsize=tiny,textwidth=2cm]{todonotes}

\newcommand{\vvcom}[1]{\todo[color=green!25!white]{victor: #1}}

\usepackage{multirow}
\usepackage[noend]{algpseudocode}
\usepackage{algorithm,algorithmicx}
\usepackage{xcolor}

\newlength{\algofontsize}
\hypersetup{
	colorlinks,
	linkcolor={red!50!black},
	citecolor={blue!50!black},
	urlcolor={blue!80!black}
}

\begin{document}
	\algrenewcommand\algorithmicrequire{\textbf{Input:}}
	\algrenewcommand\algorithmicensure{\textbf{Output:}}
	
	\title{On the extension complexity of scheduling
		\thanks{This work was partially supported by GA \v{C}R project 17-09142S and FONDECYT Regular grant 1170223.}
	 }
	
	\author{
	Hans Raj Tiwary \thanks{Department of Applied Mathematics, Charles University. {\tt hansraj@kam.mff.cuni.cz}}
	\and 			
	Victor Verdugo \thanks{Department of Mathematics, London School of Economics and Political Science. {\tt v.verdugo@lse.ac.uk}}
	\thanks{Institute of Engineering Sciences, Universidad de O'Higgins. {\tt victor.verdugo@uoh.cl}}
	\and 
	Andreas Wiese \thanks{Department of Industrial Engineering, Universidad de Chile. {\tt awiese@dii.uchile.cl}}
	}
\date{\vspace{-1em}}
\maketitle
\begin{abstract}
Linear programming is a powerful method in combinatorial optimization
with many applications in theory and practice. 
For solving a linear program quickly
it is desirable to have a formulation of small size for the given
problem. 
A useful approach for this is the construction of an {\it extended formulation}, which
is a linear program in a higher dimensional space whose projection yields the original linear program.
For many problems it is known that a small extended formulation cannot
exist. However, most of these problems are either $\mathsf{NP}$-hard (like TSP), or only
quite complicated polynomial time algorithms are known for them (like
for the matching problem). 
In this work we study the {\it minimum makespan
problem} on identical machines in which we want to assign a set of
$n$ given jobs to $m$ machines in order to minimize the maximum
load over the machines. 
We prove that the canonical formulation for this
problem has extension complexity $2^{\Omega(n/\log n)}$, even if
each job has size 1 or 2 and the optimal makespan is 2. 
This is a
case that a trivial greedy algorithm can solve optimally! 
For the
more powerful {\it configuration integer program} 
we even prove a lower bound of $2^{\Omega(n)}$. 
On the other hand, we show that there is an abstraction
of the configuration integer program admitting an extended formulation of size $f(\OPT)\cdot \text{poly}(n,m)$.
In addition, we give an $O(\log n)$-approximate integral formulation of polynomial
size, even for arbitrary processing times and for the far more general
setting of {\it unrelated machines}. 
\end{abstract}

\newpage


\section{Introduction}\label{sec:intro}


In order to solve a linear program quickly one is interested in a formulation
with as few variables and constraints as possible. A useful technique
for this are extended formulations.
A polytope $Q$ is said
to be an {\em extended formulation} or {\em extension} of a
polytope $P$ if $P$ is a linear projection of $Q$. 
There are many examples
of polytopes $P$ that require many constraints to be described but
that admit extended formulations that are much smaller. 
For instance,
the convex hull of all characteristic vectors of spanning trees in
a graph with $n$ vertex needs $2^{\Omega(n)}$ inequalities to be
described~\cite{edmonds1971matroids}, but it admits an extended
formulation of size $O(n^{3})$~\cite{martin1991using}. The \emph{extension
complexity $\xc (P)$} of a polytope $P$ is the minimum number of inequalities needed
to describe an extended formulation of it, see \cite{anor/ConfortiCZ13,Kaibel11,VanderbeckW10,Wolsey11}
for surveys on the topic.
We study the classical scheduling problem of assigning jobs on identical
machines to minimize the makespan, also known as $P||C_{\max}$ in the scheduling literature. 
We are given a set $J$ of $n$ jobs and a set $M$ of $m$ identical machines
and every job $j\in J$ has a processing time $p_{j}\in\mathbb{N}$.
The goal is to assign each job on a machine in order to minimize the maximum load over the machines, 
where the load of a machine is the sum of the processing times of the jobs assigned to it.


\subsection{Our Contribution}\label{sec:contribution}

A natural formulation for $P||C_{\max}$, known as the {\it assignment integer program}, uses a variable $x_{ij}$
for each combination of a machine $i$ and a job $j$, modelling whether
$j$ is assigned to $i$. 
The makespan is then modeled by an additional
variable $T$.
Then, its linear relaxation is given by
\begin{alignat}{2}
\min \quad \quad \quad \; T &  & &\label{eq:LP1} \\ 
\mathrm{s.t.}\quad \sum_{i\in M}x_{ij}& =1 &  & \quad\text{ for all }j\in J,\label{eq:job-equality-1-1-1}\\
\sum_{j\in J}x_{ij}p_{j}  & \le T &  &\quad  \text{ for all }i\in M,\\
x_{ij} & \ge0 &  & \quad\text{ for all }i\in M,\;\text{for all }j\in J. \label{eq:LP4}
\end{alignat}
\noindent{\it Lower bounds on the extension complexity.} We prove that there are instances with $O(n)$ jobs and machines such
that the convex hull $P_{I}$ of all integral solutions to the above
linear program has an extension complexity of $2^{\Omega(n/\log n)}$. 
The optimal
solutions form a face of $P_{I}$ and our bound also holds for this
face and hence for all
polytopes containing it as a face. Our instances satisfy
that $p_{j}\in\{1,2\}$ for each job $j\in J$ and the optimal makespan
is 2. 
Such instances can be solved optimally by a simple greedy algorithm in time $O(n+m)$.
Our key insight is that there are faces of $P_{I}$ in which
some jobs cannot be assigned to certain machines, e.g., defined via
equalities of the form $x_{ij}=0$. Hence, an extended formulation
of $P_{I}$ also yields such a formulation for the polytope of any
instance of the \emph{restricted assignment problem }in which we have
the same input as in $P||C_{\max}$ and additionally for each job
$j$ there is a set of machines $M_{j}$ such that $j$ must be assigned
to a machine in $M_{j}$. 
Using a result in \cite{goeoes2018extension} we show that
there are 3-Bounded-2-SAT instances
such that the polytope describing all feasible solutions for them has extension
complexity $2^{\Omega(n/\log n)}$. 
This might be of independent interest, in particular since 2-SAT can be solved easily in polynomial time.
We reduce these instances to the
restricted assignment problem using a reduction from 3-Bounded-2-SAT
to the restricted assignment problem such that $p_{j}\in\{1,2\}$
for each job $j\in J$ and the optimal makespan is 2~\cite{ebenlendr2014graph}. 
Hence, the polytope of all optimal
solutions has extension complexity $2^{\Omega(n/\log n)}$. By
the above this also holds for the polytope of the corresponding instance
of $P||C_{\max}$ where we ignore the sets $M_{j}$. Moreover, we show that
this holds also for each subpolytope containing all optimal solutions.

Then we consider the variation of the above formulation in which we
omit the variable $T$ and the second constraint. The set of integral
points is simply the set of \emph{all }schedules and the resulting
(polynomial size) LP is integral. 
Note that hence in this space there
exists an integral linear program of small complexity containing all optimal schedules.
However, we show that if we restrict this polytope to the 
convex hull
of all
\emph{optimal }schedules then resulting polytope again has extension
complexity $2^{\Omega(n/\log n)}$. We show that our bounds are almost
tight by giving an extended formulation of size $2^{O(n)}m$ for the
convex hull of all optimal solutions, also for arbitary processing
times and for the formulation in which the makespan $T$ is a variable. 
\\

\noindent{\it Approximate schedules: Lower bound. }For $P||C_{\max}$ there is a polynomial time $(1+\epsilon)$-approximation algorithm known
and even an \mbox{EPTAS, see e.g.,~\cite{hochbaum1987using,Jansen2016}}. Therefore,
one might wonder whether we can obtain 
a polynomial size extended
formulations that contains all optimal schedules and possibly also some $\alpha$-approximate schedules, e.g., for $\alpha=1+\epsilon$.
We show that this is not possible for any $\alpha<\frac{3}{2}$.
Even more, if we ask for a small polytope that contains \emph{all} $\alpha$-approximate schedules for some value $\alpha$ we show that this is 
does not exist for any $\alpha\le m^{1-\epsilon}$ with $\epsilon>0$. Moreover, this is tight in the sense that the set
of all $m$-approximate schedules is simply the set of all schedules
which admits the polynomial size formulation mentioned above. \\ 

%
%

\noindent{\it Approximate schedules: Upper bound.} Despite this negative result, we show that there is a polynomial
size formulation for a polytope that contains all optimal schedules
and \emph{some }$O(\log n)$-approximate schedules if the target makespan $T$ is fixed. 
Our construction
works for arbitrary job processing times and even in the more complex
setting of unrelated machines, i.e., $R||C_{\max}$, where for each
combination of a machine $i$ and a job $j$ there is a value $p_{ij}\in\mathbb{N}\cup\{\infty\}$
denoting the processing time of $j$ when it is assigned to machine
$i$. %
Key to our extended formulation is to construct an instance of bipartite
matching in which there is a vertex for each job which can be matched
to vertices representing slots on the machines. 
Then we prove that in the space with the makespan $T$ as a variable such a formulation cannot
exist which yields a separation between the two spaces even though they might appear very similar
at first glance.
\\

\noindent{\it Configuration integer program.} Finally, we study the \emph{configuration integer program} which is a popular approach
in scheduling, e.g.,~\cite{Svensson2012,Jansen2017}, with connections
to bin packing, see e.g.,~\cite{hoberg2017logarithmic,rothvoss2013approximating}.
There is a variable $y_{iC}$ for each combination
of a machine $i$ and a configuration $C\in\C(T)$ where $\C(T)$,
contains all sets of jobs whose total processing size does not exceed
an upper bound $T$ on the optimal makespan. While for large $T$
already the number of variables can be exponential, for constant $T$
this number is only polynomial and could potentially admit a small
extended formulation. However, we prove that already for the case
that $p_{j}=1$ for each job $j$ and $T=2$ there are instances with
$O(n)$ jobs and machines such that this linear program has extension complexity
$2^{\Omega(n)}$. 
To show this, we establish the maybe surprising
connection that there are such instances for which the corresponding polytope is an
extended formulation of the perfect matching polytope in a graph with
$n$ vertices and the latter has extension complexity $2^{\Omega(n)}$~\cite{Rothvoss2017}.
On the other hand, there is an abstraction of the configuration integer program
which instead of assigning a configuration $C\in\C(T)$ to each machine
$i$ only assigns a pattern that describes how many jobs of
each size are assigned to $i$ but without specifying the actual jobs.
We prove that in contrast to the configuration integer program this abstraction
admits an extended formulation of size $O(f(T)\cdot \mathrm{poly}(n,m))$
for some function $f$. 

\subsection{Related work}

\noindent{\it Lower bounds.} There are many examples known of polytopes that
do not admit small extended formulations, i.e., formulations
of polynomial size. For instance, Yanakakis~\cite{yannakakis1991expressing}
proved that for TSP there can be no such formulation that is symmetric,
i.e., stays invariant under permutation of cities. Recently, it became
an active field of research to prove such lower bounds. For instance,
Fiorini et al.~\cite{FioriniMPTW15} extended the above result
for TSP to arbitrary (possibly non-symmetric) formulations and Avis et al.~\cite{AvisT15}
showed that neither for 3-SAT, subset sum, 3D-matching, nor for MaxCut
for suspensions of cubic planar graphs,
there can be small extended formulations. Note that all these problems
are $\mathsf{NP}$-hard and hence a polynomial size extended formulation
for any of them would be very surprising. \\

\noindent{\it Easy problems with large extension complexity.} There are only few problems
in $\mathsf{P}$ for which the corresponding polytope is known to have large
extension complexity. The most famous example is probably the perfect
matching polytope for which Rothvoss showed in his celebrated result
that it has exponential extension complexity~\cite{Rothvoss2017}.
While the matching problem is in $\mathsf{P}$, the polynomial time algorithm
for it is highly complicated. 
Also, Rothvoss showed that there exists a family of matroids whose associated polytopes have exponential extension complexity~\cite{rothvoss2013some}. 
This contrasts with the fact that we can optimize over any matroid in polynomial time using the greedy algorithm~\cite{Schrijver03}.

\section{Extension complexity: Lower bound}

Suppose that we are given an instance $(J,M)$ of $P||C_{\max}$.
We consider the linear program defined by (\ref{eq:LP1})-(\ref{eq:LP4}) in Section~\ref{sec:contribution}.   
Denote by $P(J,M)$ the convex hull of all its integral solutions.
In the remainder of this section we prove the following theorem. 
\begin{theorem}
\label{thm:makespan-large-xc}For every $n$ there exists an instance
$(J,M)$ of $P|p_{j}\in\{1,2\}|C_{\max}$ with $O(n)$ jobs, $O(n)$
machines, and $\OPT(J,M)=2$ such that $\mathrm{xc}(P(J,M))\ge2^{\Omega(n/\log n)}$. 
\end{theorem}

Let $n\in\mathbb{N}$. For any given SAT formula $\Phi$ with $n$
variables we define the polytope $\mathrm{SAT}(\Phi)$ as the convex
hull of all satisfying assignments, i.e., $\mathrm{SAT}(\Phi):=\mathrm{conv}(\{y\in\{0,1\}^{n}:\Phi(y)=1\})$.
We use the following theorem that follows easily from~\cite{goeoes2018extension}. 
\begin{theorem}
\label{thm:formula-large-xc}For every $n\in\mathbb{N}$ there exists
a 2-SAT formula $\Phi$ with $O(n)$ variables and $O(n)$ clauses
such that $\mathrm{xc}(\SAT(\Phi))\ge2^{\Omega(n/\log n)}$. Each
clause of $\Phi$ contains exactly two literals.
\end{theorem}
\begin{proof}[Proof of Theorem~\ref{thm:formula-large-xc}]
Let $n\in\mathbb{N}$. In~\cite{goeoes2018extension}
it is shown that there exists a graph $G=(V,E)$ with $n$ vertices
such that for its independent set polytope $P_{G}$ it holds that
$\xc(P_{G})\ge2^{\Omega(n/\log n)}$, i.e., $P_{G}\subseteq[0,1]^{n}$
is the convex hull of all incidence vectors of independent sets of
$G$. Moreover, the degree of $G$ is bounded by a global constant
that is independent of $n$. Based on $G$ we construct a 2-SAT formula
$\Phi$. For each node $v\in V$ we introduce a variable $x_{v}$,
the intuition being that $x_{v}$ is true if $v$ is in the independent
set. For each edge $\{u,v\}\in E$ we introduce a clause $(\neg x_{v}\vee\neg x_{u})$,
modelling that not both $u$ and $v$ can be in the independent set.
The number of variables is $n$ and since $G$ has bounded degree
each variable appears in at most $O(1)$ clauses. Hence, the number
of clauses is also $O(n)$. Each satisfying assignment of $\Phi$
corresponds to an independent set of $G$ and vice versa. Therefore,
$\mathrm{xc}(\SAT(\Phi))\ge2^{\Omega(n/\log n)}$.
\end{proof}

Let $\Phi$ denote the formula due to Theorem~\ref{thm:formula-large-xc}.
We transform $\Phi$ into an equivalent 3-Bounded-2-SAT formula $\Phi'$
using a standard reduction
i.e., $\Phi'$ is a 2-SAT formula in which each variable appears at
most three times. Let $x_{i}$ be a variable in $\Phi$ and assume
that $x_{i}$ occurs $k$ times. We introduce $k$ new variables $x_{i}^{(1)},...,x_{i}^{(k)}$
and for each $\ell\in[k]$ we replace the $\ell$-th occurrence of $x_{i}$
with $x_{i}^{(\ell)}$. Additionally, we add the clauses $\Big(\neg x_{i}^{(\ell)}\vee x_{i}^{(\ell+1)}\Big)$
for each $\ell\in\{1,...,k-1\}$
and the clause $\Big(\neg x_{i}^{(k)}\vee x_{i}^{(1)}\Big)$. Hence, in any satisfying
assignment of the resulting formula, either $x_{i}^{(\ell)}=1$ for
each $\ell\in[k]$ or $x_{i}^{(\ell)}=0$ for each $\ell\in[k]$. We do
this transformation with each variable in $\Phi$. Let $\Phi'$ denote
the resulting formula. By construction, we have that $\Phi'$ has
$O(n)$ variables and $O(n)$ clauses, using that both quantities
are linear in the number of literals in $\Phi$ and the latter is
bounded by $O(n)$. By construction, each variable appears exactly
three times and at most two times positively and at most two times
negatively. Also, each clause contains exactly two literals.

\subsection{Reduction to the restricted assignment problem}

\noindent{\it Construction.} Next, based on $\Phi'$ we construct an instance of the restricted
assignment problem. We invoke the reduction from~\cite{ebenlendr2014graph}.
For each variable $x$ in $\Phi'$ we introduce a machine $i(x)$,
a machine $i(\neg x)$, and a job $j(x)$. The job $j(x)$ has processing
time $p_{j(x)}=2$ and it can be assigned only on $i(x)$ and $i(\neg x)$,
i.e., $M_{j(x)}=\{i(x),i(\neg x)\}$. 
The intuition behind is that if $j(x)$
is scheduled on $i(x)$ then $x$ is true and if $j(x)$ is scheduled
on $i(\neg x)$ then $x$ is false. For each clause $c$ we introduce
one machine $i(c)$. For each variable $x$ that occurs in $c$ we
introduce a job $j(c,x)$ with $p_{j(c,x)}=1$. If $x$ occurs positively
in $c$ then we define $M_{j(c,x)}=\{i(c),i(\neg x)\}$, otherwise
we define $M_{j(c,x)}=\{i(c),i(x)\}$. Finally, for each clause $c$
we introduce a job $j(c)$ with $p_{j(c)}=1$ and $M_{j(c)}=\{i(c)\}$.

Since the total number of variables and clauses is $O(n)$, we introduced
$O(n)$ jobs and machines. 
Let $J'$ denote
the set of jobs and let $\bar{m}$ denote the number of machines
defined so far. 
We want that in solutions with makespan 2 each machine has a load of exactly 2. 
To this end,
we introduce a set $D$ of $2\bar{m}-\sum_{j\in J'}p_{j}$ dummy jobs of length
1 each. 
For each dummy job $j\in D$, it can be assigned to any machine, i.e.,
$M_{j}=M$.\\

\noindent{\it Correctness.} For each satisfying assignment of $\Phi'$ there is a schedule of
makespan~2: If a variable $x$ is assigned to be $x=1$ in the satisfying
assignment then we schedule $j(x)$ on machine $i(x)$, otherwise
we schedule $j(x)$ on machine $i(\neg x)$. Consider a clause $c$.
There must be at least one variable that satisfies $c$. For each
such variable $x$, if $x$ occurs positively in $c$ then we assign
$j(c,x)$ on $i(\neg x)$, otherwise we assign $j(c,x)$ on $i(x)$.
For each variable $y$ that does not satisfy $c$ we assign $j(c,y)$
to $i(c)$. Also, we assign $j(c)$ to $i(c)$. Using that each variable
appears at most twice positively and at most twice negatively, one
can check that the resulting makespan is 2. Finally, we assign the
dummy jobs to the machines such that each machine still has a makespan
of at most 2. 
This is the optimal solution since the largest processing time is 2.
One can also easily show that if there is a solution of makespan 2
then there exists a satisfying assignment for $\Phi'$, see~\cite{ebenlendr2014graph}
for details.

\paragraph*{Faces of scheduling polytope.}

Let $J$ and $M$ denote the set of jobs and machines in the construction
above, respectively. Also, let $M_{j}$ denote the set of allowed
machines for each job $j\in J$. We consider the polyhedron $P(J,M)$.
Note that $P(J,M)$ ignores the sets $M_{j}$ of allowed machines
for each job $j$. We argue that there is a face $P'(J,M)$ of $P(J,M)$
such that each vertex of $P'(J,M)$ corresponds to a schedule in which
each job $j\in J$ is assigned on a machine in $M_{j}$ and the makespan
is 2. Observe that the inequalities $\sum_{j\in J}\sum_{i\in M\setminus M_{j}}x_{ij}\ge0$
and $T\ge2$ are valid inequalities for $P(J,M)$. Hence, the set
\[P'(J,M)=P(J,M)\cap\left\{(x,T):\sum_{j\in J}\sum_{i\in M\setminus M_{j}}x_{ij}=0\text{ and }T=2\right\}\]
is a face of $P(J,M)$. Also, $\mathrm{xc}(P'(J,M))\le\mathrm{xc}(P(J,M))$.

\begin{proof}[Proof of Theorem~\ref{thm:makespan-large-xc}]
We describe now a linear projection $f:P'(J,M)\rightarrow\mathrm{SAT}(\Phi')$
of $P'(J,M)$ to $\SAT(\Phi')$. Given a point $x\in P'(J,M)$ we
define that the component of $f(x)$ corresponding to the variable
$y_{\ell}$ equals to $x_{i(y_{\ell}),j(y_{\ell})}$. By the construction
above and the proof of correctness of the reduction, for each integral
point $y\in\mathrm{SAT}(\Phi')$ there exists a feasible schedule
for $(J,M)$ with makespan 2 in which each job $j$ is assigned to
a machine in $M_{j}$. Additionally, for each variable $y$ in $\Phi'$
it holds that if $y=1$ then $j(y)$ is assigned on machine $i(y)$ in this schedule.
Thus, there exists an integral point $x\in P'(J,M)$ such that $f(x)=y$.
Similarly, for each integral point $x\in P'(J,M)$ we have that $f(x)\in\mathrm{SAT}(\Phi')$.
Thus, $f(P'(J,M))=\mathrm{SAT}(\Phi')$ and therefore, $\xc(\SAT(\Phi'))\le\mathrm{xc}(P'(J,M))$.
Finally we give a linear projection $g:\SAT(\Phi')\rightarrow\SAT(\Phi)$.
Recall that for each variable $x_{i}$ in $\Phi$ we introduced a
set of new variables $x_{i}^{(1)},...,x_{i}^{(k)}$ in $\Phi'$. By
construction, in each satisfying assignment for $\Phi'$ all these
variables $x_{i}^{(1)},...,x_{i}^{(k)}$ have the same value. Therefore,
we define that the component of $g(x)$ corresponding to $x_{i}$
equals $x_{i}^{(1)}$ for each variable $x_{i}$ in $\Phi$. We obtain
$g(\mathrm{SAT}(\Phi'))=\mathrm{SAT}(\Phi)$. Therefore, $\xc(\SAT(\Phi))\le\xc(\SAT(\Phi'))$.
Hence, $\mathrm{xc}(P(J,M))\ge\mathrm{xc}(P'(J,M))\ge\xc(\SAT(\Phi'))\ge\xc(\SAT(\Phi))\ge2^{\Omega(n/\log n)}$.
This completes the proof of Theorem~\ref{thm:makespan-large-xc}.
\end{proof}
Since already the face $P'(J,M)$ of $P(J,M)$ containing the optimal solutions
to $(J,M)$ has extension complexity $2^{\Omega(n/\log n)}$
we obtain the following corollary.%

\begin{corollary}
\label{cor:only-optimal-points}For every $n$ there exists an instance
$(J,M)$ of $P|p_{j}\in\{1,2\}|C_{\max}$ with $O(n)$ jobs, $O(n)$
machines, and $\OPT(J,M)=2$ such that for any integral polyhedron
$\bar{P}(J,M)\subseteq P(J,M)$ that contains all optimal and possibly
also some other solutions to $(J,M)$ it holds that $\mathrm{xc}(\bar{P}(J,M))\ge2^{\Omega(n/\log n)}$. 
\end{corollary}

\subsection{Extensions}

Above we proved that $P(J,M)$ has large extension complexity. The
polyhedron $P(J,M)$ is defined using the variable $T$ which represents
an upper bound on the makespan of the respective solution. This raises
the question whether there exist compact extended formulations in
the space defined only via the variables $x$, without the variable
$T$. 
Formally, we consider the polyhedron 
\[Q(J,M)=\mathrm{conv}\left(\left\{ x\in \{0,1\}^{M\times J}:\sum_{i\in M}x_{ij}=1\,\text{ for all } j\in J\right\} \right)\]
describing the set of all schedules, including all optimal ones. %
One can easily show that its extension complexity is $O(nm)$ by simply
taking its linear relaxation. Observe that Corollary~\ref{cor:only-optimal-points}
rules out such a polyhedron in the space lifted with the variable $T$.
\begin{proposition}
\label{lem:comlexity-Q(J,M)}The polyhedron $Q(J,M)$ has extension
complexity $O(nm)$. 
\end{proposition}
%
This follows by observing that the extreme points of the linear relaxation are integral. 
Let $x$ be a feasible point in the relaxation such that $x_{ij}\in(0,1)$ for some machine $i\in M$ and some job $j\in J$. 
Due
to Equality~(\ref{eq:job-equality-1-1-1}) there must be another
machine $\ell\ne i$ such that $x_{\ell j}\in(0,1)$. For sufficiently
small $\epsilon>0$ we define a new solution $\tilde x =x+\varepsilon (\mathds{1}_{ij}-\mathds{1}_{\ell j})$,
where $\mathds{1}_{a}$ is the canonical vector with value $1$ at entry $a$ and zero otherwise.
Similarly, we define a new solution $\hat x=x+\varepsilon (\mathds{1}_{\ell j}-\mathds{1}_{i j})$.
We can find $\varepsilon>0$ such that $\tilde x,\hat x \ge0$ and
$x=\frac{1}{2}(\tilde x +\hat x)$, and hence $x$ is not an extreme point.
Therefore,
the extension complexity of $Q(J,M)$ is $O(nm)$.\\
%

However, when we minimize the makespan, we are interested only in
the set of all integral points in $Q(J,M)$ that correspond to optimal
solutions. Note that their convex hull 
corresponds
to the face of $P(J,M)$ containing all points $(x,T)$ with $T=\OPT$.
Due to Corollary~\ref{cor:only-optimal-points} this polytope has
extension complexity~$2^{\Omega(n/\log n)}$. We can strengthen this
statement. In the instance $(J,M)$ constructed above the optimal
makespan is 2 and any non-optimal solution has a makespan of at least
3. Thus, all integral polytopes $\bar{Q}(J,M)\subseteq Q(J,M)$ that
contain all optimal solutions to $(J,M)$ and possibly some $(3/2-\epsilon)$-approximate
solutions have large extension complexity. Note that this contrasts
the fact that $P||C_{\max}$ admits an EPTAS~\cite{Jansen2016}.
\begin{corollary}
\label{cor:large-xc-apx}Let $\epsilon>0$ with $\epsilon\le1/2$.
For every $n$ there exists an instance $(J,M)$ of $P|p_{j}\in\{1,2\}|C_{\max}$
with $O(n)$ jobs, $O(n)$ machines, and $\OPT(J,M)=2$ such that for
any integral polytope $\bar{Q}(J,M)\subseteq Q(J,M)$ whose vertices
consist of all optimal solutions and possibly some $(\frac{3}{2}-\epsilon)$-approximate
solutions to $(J,M)$ it holds that $\mathrm{xc}(\bar{Q}(J,M))\ge2^{\Omega(n/\log n)}$. 
\end{corollary}

Next, we show that there cannot be an integral polytope of polynomial
extension complexity containing \emph{all }$\alpha$-approximate
solutions of $Q(J,M)$ for any $\alpha\le n^{1-\epsilon}$. 
\begin{corollary}
\label{cor:large-xc-apx-all}For every $n\in\mathbb{N}$ and every
$\alpha\ge1$ there exists an instance $(J',M')$ of $P|p_{j}\in\{1,2\}|C_{\max}$
with $O(\alpha n)$ jobs and $O(\alpha n)$ machines such that for
any integral polytope $\tilde{Q}(J,M)\subseteq Q(J,M)$ whose vertices
are all $\alpha$-approximate solutions to $(J',M')$ it holds that
$\mathrm{xc}(\tilde{Q}(J,M))\ge2^{\Omega(n/\log n)}$. In particular,
if $\alpha=n^{1/\epsilon}$ for some $\epsilon>0$ then the instance
has $\bar{n}=O(n^{1+1/\epsilon})$ jobs and machines and $\mathrm{xc}(\tilde{Q}(J,M))\ge2^{\bar{n}^{\Omega(\epsilon)}}$. 
\end{corollary}

\begin{proof}
Assume by contradiction that
such a polytope exists and suppose that $\alpha\in\mathbb{N}$. Take
the instance $(J,M)$ defined above. We add $(2\alpha-2)|M|$ dummy
jobs of length 1 each, denote them by $\Jd$. Also, we add $(\alpha-1)|M|$
dummy machines, denote them by $\Md$. Note that the optimal makespan
is still 2 and thus any $\alpha$-approximate solution has makespan
$2\alpha$. Let $(J',M')$ denote the resulting instance and observe
that it has $O(\alpha n)$ jobs and machines. We consider $\tilde{Q}(J',M')$.
Then for each dummy machine $i\in\Md$ the inequality $\sum_{j\in J'}x_{ij}\ge0$
is a valid inequality. Also, consider an assignment of the jobs in
$\Jd$ to the machines $M'\setminus\Md$ such that each machine $i\in M'\setminus\Md$
gets a load of $2\alpha-2$ in this way. Formally, we define a map
$g:\Jd\rightarrow M'\setminus\Md$ such that $|g^{-1}(i)|=2\alpha-2$
for each $i\in M'\setminus\Md$. Then for each job $j\in\Jd$ the
inequality $x_{g(j)j}\le1$ is a valid inequality. Therefore, the set $\tilde{Q}'(J',M')$ defined as
\[\tilde{Q}(J,M)\cap\left\{x:\sum_{j\in J'}x_{ij}=0\,\,\text{ for all } i\in\Md\text{ and } x_{g(j)j}=1\,\,\text{ for all } j\in\Jd\right\}\]
is a face of $\tilde{Q}(J,M)$. Therefore, $\tilde{Q}'(J',M')$ contains
exactly the solutions in which the dummy jobs are assigned as described
by the map $g$ and the non-dummy jobs $J=J'\setminus\Jd$ are assigned
such that they give a load of 2 on each machine. Thus, there is a
linear projection of $\tilde{Q}'(J',M')$ to the polytope $\bar{Q}(J,M)$
as defined in Theorem~\ref{thm:formula-large-xc}. This implies that
$\xc(\tilde{Q}(J',M'))\ge\xc(\tilde{Q}'(J',M'))\ge\xc(\bar{Q}(J,M))\ge2^{\Omega(n/\log n)}$.
\end{proof}


Finally, we can show that our lower bounds from Theorem~\ref{thm:makespan-large-xc}
and Corollary~\ref{cor:large-xc-apx} are almost tight by giving
an upper bound of $2^{O(n)}m$. 
%
%
%
\begin{theorem}
\label{thm:upper-bound}Let $(J,M)$ be an instance of $P||C_{\max}$
with $n$ jobs and $m$ machines. Let \textup{$\bar{Q}(J,M)$} denote
the convex hull of the vertices corresponding to all optimal solutions
to $(J,M)$ in \textup{$Q(J,M)$. }It holds that $\xc(\bar{Q}(J,M))\le2^{O(n)}m$
and $\xc(P(J,M))\le2^{O(n)}m$. 
\end{theorem}

In the following we prove Theorem~\ref{thm:upper-bound}. 
Given an instance $(J,M)$,
we first describe
a dynamic program that computes a solution to $(J,M)$ in time $2^{O(n)}m$
for some given target makespan $T$, assuming that such a solution
exists. 
Then based on it we define an extended formulation of $\bar{Q}(J,M)$
of size $2^{O(n)}m$. \\

\noindent{\it Dynamic program.}  We introduce a cell $(J',m',T)$ for each
subset $J'\subseteq J$ of jobs and each integer $m'\in \{1,\ldots,m\}$.
In $(J',m',T)$ we want to store a schedule for the jobs $J'$ on
the machines $\{m',\ldots,m\}$ such that each machine has a makespan of
at most $T$, assuming that such a schedule exists. 
Consider the case $m'=m$, that is, we look for a schedule in the single machine $m$. 
For each set of jobs $J'$ with $p(J')\le T$ we store
in the cell $(J',m,T)$ the schedule that assigns all jobs in $J'$
to machine $m$. For each set of jobs $J'$ with $p(J')>T$ we store
$\bot$ in the cell $(J',m,T)$, indicating that no feasible schedule
exists for $(J',m,T)$. 
Now suppose we are given a cell $(J',m',T)$ with
$m'<m$. If there is a subset $J''\subseteq J'$ with $p(J'')\le T$
such that in the cell $(J'\setminus J'',m'+1,T)$ we stored a schedule
(and not $\bot$) then in $(J',m',T)$ we store the schedule that
assigns $J''$ on machine $m'$ and the schedule in the cell $(J'\setminus J'',m'+1,T)$
on the machines $\{m'+1,\ldots,m\}$. If no such set $J''\subseteq J'$ exists
then we store $\bot$ in $(J',m',T)$. Finally, if we stored a schedule
$(J,1,T)$ then we output this schedule, otherwise $(J,1,T)$ contains
$\bot$ and we output that there is no schedule for $J$ on $m$ machines
with makespan at most $T$. 
This dynamic program table has $2^{n}m$ cells
and evaluating one cell takes $2^{n}$ time, which yields a total
running time of $2^{O(n)}m$.

In the transition above, for each cell $(J',m',T)$ we took an \emph{arbitrary
}subset $J''$ such that $p(J'')\le T$ and $(J'\setminus J'',m'+1,T)$
does not contain $\bot$. In fact, we can construct \emph{any }schedule
of makespan at most $T$ in this way if for each cell $(J',m',T)$
we choose for $J''$ the set of jobs that the respective schedule
assigns on machine $m'$. In the sequel, we define a graph $G$
with $2^{O(n)}m$ vertices including two special vertices $s$ and $t$ such
that any path from $s$ to $t$ corresponds to a solution that the
above dynamic program might compute for suitable choices for $J''$. 
Then we define
a linear program whose vertices are exactly these paths which then
yield an extended formulation to $\bar{Q}(J,M)$ if we choose $T=\OPT$.\\

\noindent{\it Construction of the graph.} Let $T\ge \OPT$.
Let $\mathcal{V}_T$ be the set of cells $(J',m',T)$ of the dynamic program table such that there exists a schedule for the jobs $J'$
on machines $\{m',\ldots,m\}$ of makespan at most $T$. 
For convenience we add a dummy element $(\emptyset,m+1,T)$ to $\mathcal{V}_T$.
Then, in our graph $G$ the set of vertices corresponds to $\mathcal{V}_T$.
For each $J',\bar{J}'\subseteq J,m'\in\{1,...,m\}$
we introduce an arc $((J',m',T),(\bar{J}',m'+1,T))$ if and
only if $(J',m',T),(\bar{J}',m'+1,T)\in \mathcal{V}_T$ and there
is a set $J''\subseteq J'$ with $p(J'')\le T$ such that $\bar{J}'=J'\setminus J''$.
Let $s=(J,1,T)$ be the {\it source} and $t=(\emptyset,m+1,T)$ the {\it sink}. 
We call $\mathcal{A}_T$ the set of arcs, and then $G=(\mathcal{V}_T,\mathcal{A}_T)$.

\begin{lemma}
\label{lem:shortest-path}
Every $s$-$t$ path in the graph $G$ corresponds to a schedule $S$
of makespan at most $T$. 
Furthermore, every schedule $S$ of makespan at most $T$ induces an $s$-$t$ path in $G$.
\end{lemma}

\begin{proof}
Let $P$ be an $s$-$t$ path with vertices
$(J'_{1},1,T),(J'_{2},2,T),...,(J'_{m},m,T),(J'_{m+1},m+1,T),$
where $J'_1=J$ and $J'_{m+1}=\emptyset$.
Then, our construction of $G$ guarantees that $p(J'_{\ell+1}\setminus J'_{\ell})\le T$
for each $\ell\in\{1,...,m\}$. 
Thus, for each arc $((J'_{\ell},\ell,T),(J'_{\ell+1},\ell+1,T))$
we assign the jobs in $J'_{\ell}\setminus J'_{\ell+1}$ to machine $\ell$ which yields
a schedule of makespan at most $T$. 
Conversely,~consider a schedule
$S$ of makespan at most $T$ and for each $i\in \{1,\ldots,m\}$ let $J_{i}\subseteq J$ be the jobs that are assigned to machine $i$ in $S$. 
Then,
there is a path that for each $i\in\{1,...,m+1\}$
visiting the vertex $(\bigcup_{\ell=i}^{m}J_{\ell},i,T)$. 
Observe that for each $i\in \{1,\ldots,m\}$ we~have that $((\bigcup_{\ell=i}^{m}J_{\ell},i,T),(\bigcup_{\ell=i+1}^{m}J_{\ell},i+1,T))\in \mathcal{A}_T$ since $\bigcup_{\ell=i}^{m}J_{\ell}\setminus \bigcup_{\ell=i+1}^{m}J_{\ell}=J_{i}$~and $p(J_{i})\le T$.
\end{proof}

\noindent{\it Extended formulation.} 
We define a linear program with a variable $y_a$ for
each arc $a\in \mathcal{A}_T$. We add constraints
that describe a flow in $G$ such that we send exactly one unit of
flow from $s$ to $t$. In particular, we require that one unit of
flow leaves $s$, one unit of flow enters $t$, and on all other vertices
there is flow conservation. 
For every job $j\in J$, let $\mathcal{J}^{T}_{ij}=\{((J',i,T),(\bar{J'},i+1,T))\in \mathcal{A}_T:j\in J'\setminus\bar{J}'\}$
which are the arcs such that if one of them is contained in an $s$-$t$path $P$ then in the schedule corresponding to $P$ job $j\in J$ is
assigned on machine $i\in M$.
Then, consider the following linear program
\begin{alignat}{2}
	 \sum_{a\in \delta^+(v)}y_a-\sum_{a\in \delta^-(v)}y_a                    & =   \begin{cases} 0 & \text{ for all}\; v \in \mathcal{V}_T\setminus \{s,t\},\\1 & \text{ if }v=s,\\ -1 & \text{ if }v=t.  \end{cases} && \label{eq:shortest}\\
  \sum_{a\in \mathcal{J}_{ij}^T}y_a & =  x_{ij} \,\,\,\text{ for all } i\in M,\,\text{ for all } j\in J,\label{eq:ext-projection}\\
  x_{ij} & \ge0 \,\,\,\,\,\,\,\text{ for all } i\in M,\,\text{ for all } j\in J,\\
  y_a & \ge0 \,\,\,\,\,\,\thinspace\text{ for all }a\in \mathcal{A}_T, \label{eq:non-neg-shortest}
\end{alignat}
where for each vertex $v$, 
we have that $\delta^+(v),\delta^-(v)$ denote the out-going and in-going edges at $v\in \mathcal{V}_T$, respectively.
\begin{proof}[Proof of Theorem~\ref{thm:upper-bound}]
For $T=\OPT$, we project the polytope of the linear program above
to $\bar{Q}(J,M)$
by defining $x_{ij}$ according to \eqref{eq:ext-projection} for each job $j\in J$ and each machine $i\in M$.
By Lemma~\ref{lem:shortest-path} this yields an extended formulation of size $2^{O(n)}m$. 
We describe now how to extend the above to an extended formulation
for $P(J,M)$ of size $2^{O(n)}m$. 
For that, we show how to modify the dynamic program used for constructing the extended formulation of $\bar{Q}(J,M)$.

Note that for the makespan of
any feasible solution there are only $2^{n}$ options since the makespan
equals the sum of the processing times of the jobs on some machine
and there are only $2^{n}$ options for this quantity. Let $\mathcal{T}$
be a set containing all these options. We define a cell $(J',m',T)$
for each subset $J'\subseteq J$ of jobs, each integer $m'$ with
$m'\in \{1,\ldots,m\}$ and for each $T\in\mathcal{T}$. 
For each such cell in
which the previous dynamic program above does not store $\bot$ we define a vertex 
like above and also arcs 
as before. 
Note that there are no arcs between two vertices $(J',m',T)$,
$(\bar{J}',\bar{m}',\bar{T})$ with $T\ne\bar{T}$. 
We define
a new source vertex $s'$ and introduce an arc $(s',(J,1,T))$ for each
$T\in \T$, assuming that the vertex $(J,1,T)$ exists. 
Similarly, we
define a new sink vertex $t'$ and an arc $((\emptyset,m+1,T),t')$
for each $T\in\T$. 
A path from $s'$ to $t'$ represents
the choice of a makespan $T\in\T$ and defining a schedule with makespan
at most $T$. 
As before, consider a linear program computing $s$-$t$ paths in this new graph, using a flow formulation and $y$ variables as before, that is, constraints (\ref{eq:shortest}) and (\ref{eq:non-neg-shortest}).
We lift this linear program by adding a variable $\bf T$
and introducing a new constraint 
\begin{equation}{\bf T}\ge\sum_{T\in\T}T\cdot y_{(s',(J,1,T))}.\end{equation}
Then we project each point $(y,{\bf T})$ in the resulting polyhedron
to the point $(x,{\bf T})\in P(J,M)$ by setting 
$x_{ij}=\sum_{T\in\T} \sum_{a\in \mathcal{J}_{ij}^T}y_a$ for each $i\in M$ and each $j\in J$.
%
That concludes the proof of the theorem.
\end{proof}

\section{Approximate polynomial size extended formulation}
In the previous section we showed
in Corollary~\ref{cor:large-xc-apx-all} that there can be no polynomial
size extended formulation of $Q(J,M)$ whose vertices are exactly
all $\alpha$-approximate solutions, for essentially any $\alpha$.
However, in this section we show that there is a small extended formulation
that contains \emph{all} optimal schedules and \emph{some }approximate
schedules. 
Our formulation works even in the more general setting of unrelated
machines, i.e., $R||C_{\max}$, where the processing time of a job
$j$ can depend on the machine $i$ that it is assigned to and for
each such combination the input contains a value $p_{ij}\in\mathbb{N}\cup\{\infty\}$.
In this setting the polytope $Q(J,M)$ defined exactly in the same as way as before,
however, now a solution $x$ is optimal if $\sum_{j\in J}x_{ij}p_{ij}\le \OPT$
for each machine $i\in M$.\\

\noindent{\it Construction of the extended formulation.} In the sequel let $T\ge \OPT$. The intuition behind our construction
is the following. We define a bipartite graph $G$ in which we have
one vertex $v_{j}$ for each job $j\in J$ and $n$ vertices \textit{$w_{(i,1)},...,w_{(i,n)}$
}for each machine $i\in M$. 
Each vertex $w_{(i,\ell)}$ corresponds to
a \emph{slot }for machine $i$. We search for a matching that assigns
each job vertex $v_{j}$ to some slot vertex $w_{(i,\ell)}$ and we
allow this assignment if and only if $p_{ij}\le T/\ell$, i.e., we
introduce an edge $\{v_{j},w_{(i,\ell)}\}$ if and only if $p_{ij}\le T/\ell$.
The intuition is that in $\OPT$ for each machine $i$ there can be
at most one job $j$ assigned to $i$ with $p_{ij}\in(T/2,T]$, at
most two jobs $j$ with $p_{ij}\in(T/3,T]$ and more general at most
$\ell$ jobs $j$ with $p_{ij}\in(T/(\ell+1),T]$. Hence, there is
a matching in $G$ that corresponds to $\OPT$. On the other hand,
in any matching the total processing time of the jobs on any machine
$i$ is bounded by $T+T/2+T/3+...+T/n=T\cdot H_{n}=O(T\log(n))$ 
which is at most $\OPT\cdot O(\log(n))$ if $T=O(\OPT)$ (e.g., set $T=\OPT$ or
set $T$ to be the makespan found by a 2-approximation algorithm
for $R||C_{\max}$~\cite{shmoys-tardos,shchepin2005optimal}).\\

We define an integral polytope that models this bipartite matching. 
For each $(i,\ell)$
let $J_{i\ell}$ be the subset of jobs $j$ whose vertex $v_{j}$
has an edge incident to $w_{(i,\ell)}$. For each edge $\{v_{j},w_{(i,\ell)}\}$
we have a variable $y_{ji\ell}$ that intuitively indicates whether
job $j$ is assigned to the slot $(i,\ell)$. On top of this, we project
the resulting polytope to the space of the variables $x$ of
$Q(J,M)$ which yields $\bar{Q}(J,M)$.
\begin{alignat}{2}
\sum_{j\in J_{i\ell}}y_{ji\ell} & \le1 &  & \quad\text{ for all }i\in M,\;\text{for all }\ell\in\{1,\ldots,n\},\label{eq:matching1}\\
\sum_{i\in M}\sum_{\ell=1}^{n}y_{ji\ell} & =1 &  & \quad\text{ for all }j\in J,\;\text{}\\
\sum_{\ell:j\in J_{i\ell}}y_{ji\ell} & =x_{ij} &  & \quad\text{ for all }i\in M,\;\text{for all }j\in J,\\
x_{ij} & \ge0 &  & \quad\text{ for all }i\in M,\;\text{for all }j\in J,\\
y_{ij\ell} & \ge0 &  & \quad\text{ for all }i\in M,\;\text{for all }j\in J, \;\text{for all }\ell\in\{1,\ldots,n\}.
\end{alignat}

\begin{theorem}
\label{thm:ext-approximate} Given an instance $(J,M)$ of $R||C_{\max}$
there exists an extended formulation of size $O(n^{2}m)$ for an integral
polytope $\bar{Q}(J,M)\subseteq Q(J,M)$ whose vertices correspond
to all optimal solutions and some $O(\log n)$-approximate solutions
to $(J,M)$. 
\end{theorem}

\begin{proof}[Proof of Theorem~\ref{thm:ext-approximate}]
Consider an optimal schedule and let $x$ be its corresponding solution in $Q(J,M)$. 
For each machine $i\in M$, let $J_i(x)$ be the subset of job vertices $v_j$ in $G$ such that $x_{ij}=1$, that is, $J_i(x)=\{v_j:x_{ij}=1\}$, and consider $G_i$ the bipartite subgraph of $G$ that is induced by the vertices $J_i(x) \cup \{w_{(i,1)},...,w_{(i,n)}\}$.
We check that there exists a matching in $G_i$ that covers $J_i(x)$.
Consider $W\subseteq J_i(x)$.
There exists at least one job in $t\in W$ with processing time at most $T/|W|$, otherwise $x$ would exceed the makespan $T$ in machine $i$.
Therefore, job $t$ is connected to $w_{(i,\ell)}$ for every $\ell\le |W|$, and its degree in $G_i$ is at least $|W|$.
In particular, $W$ is connected to at least $|W|$ slots in the bipartite subgraph $G_i$.
By Hall's theorem we conclude that exists a matching in $G_i$ covering every job in $J_i(x)$.
For every job in $J_i(x)$ we define $y_{ji\ell}=1$ if $v_j$ is connected to $w_{(i,\ell)}$ in the matching, and $y_{ji\ell}=0$ otherwise.
By construction the solution $(x,y)$ satisfies the constraints in the program above.

On the other hand, consider a vertex solution $(x,y)$.
The integrality of $(x,y)$ comes from the fact that the linear program restricted to $y$ variables is a bipartite matching formulation.
We bound the makespan of the schedule obtained from $x$.
For each machine $i\in M$, we have
\begin{align*}
\sum_{j\in J}p_{ij}x_{ij}=\sum_{j\in J}p_{ij}\sum_{\ell:j\in J_{i\ell}}y_{ji\ell}=\sum_{\ell=1}^{n}\sum_{j\in J_{i\ell}}p_{ij}y_{ji\ell}. 
\end{align*}
For every $j\in J_{i\ell}$, we have that $p_{ij}\le T/\ell$.
This, together with constraint (\ref{eq:matching1}) allows us to upper bound the last summation above by
\[\sum_{\ell=1}^{n}\frac{T}{\ell}\sum_{j\in J_{i\ell}}y_{ji\ell}\le T\sum_{\ell=1}^{n}\frac{1}{\ell}=T\cdot H_n=O(T\log n).\qedhere\]
\end{proof}

\section{Extension complexity of configuration-LP}

An alternative way to formulate $P||C_{\max}$ as an integer linear program is
via a {\it configuration integer program}. For a given target makespan $T\ge0$ we define
the set of all \emph{configurations $\C(T)$ }to be all subsets of
jobs whose total processing time is at most $T$, i.e., $\C(T):=\{L\subseteq J:\sum_{j\in L}p_{j}\le T\}$.
For each machine $i\in M$ and each configuration $C\in\C(T)$ we
introduce a variable $y_{iC}\in\{0,1\}$ that models whether machine
$i$ gets exactly the jobs in configuration $C$ assigned to it. Denote
by $P_{\mathrm{config}}(J,M,T)$ the convex hull of all solutions
to the integer program below.
%
\begin{alignat}{2}
	 \sum_{i\in M}\sum_{C\in \C(T):j\in C} y_{iC}                       & =  1   &&\quad\text{for all } j\in J, \text{ } \\
	 	 \sum_{C\in \C(T)} y_{iC}                       & =  1   &&\quad\text{for all } i\in M, \text{ } \\
						    y_{iC}              & \in \{0,1\}   &&\quad\text{for all }i\in M, \text{ for all } C\in \C(T).
\end{alignat}

For arbitrary values of $T$ the number of variables of $P_{\mathrm{config}}(J,M,T)$
can be exponential in the input length. However, for constant $T$
there are only a polynomial number of variables (and constraints)
since the number of possible configurations is bounded by ${n \choose T}=n^{O(T)}$.
This raises the questions whether $P_{\mathrm{config}}(J,M,T)$ admits
a small extended formulation for such values of $T$. 
If $T=2$ and $p_{j} \in \{1,2\}$ for each $j\in J$ Theorem~\ref{thm:makespan-large-xc} implies a lower bound of $2^{\Omega(n/\log n)}$
since there is an easy projection of the configuration-LP to $P(J,M)$. We strengthen this to the case where $p_{j}=1$ for each $j\in J$ and to a lower bound of $2^{\Omega(n)}$.


\begin{theorem}\label{thm:config-LP}
For every $n\in\mathbb{N}$ there is an instance $(J,M,T)$ of $P||C_{\max}$
with $n$ machines and $O(n)$ jobs such that the polytope $P_{\mathrm{config}}(J,M,T)$
has an extension complexity of $2^{\Omega(n)}$. It holds that $T=2$
and $p_{j}=1$ for each $j\in J$. 
\end{theorem}
\begin{proof}[Proof of Theorem~\ref{thm:config-LP}]
For every $n\in \mathbb{N}$ we construct an instance of $(J,M,T)$ by $T=2$,
$M$ contains $n$ machines, and $J$ contains $2n$ jobs with $p_{j}=1$
for each $j\in J$. The set $\C(T)$ is thus the set of all pairs
of jobs. We show that there is a linear map that projects each solution
to $P_{\mathrm{config}}(J,M,T)$ to a matching in a complete graph
on $2n$ vertices. Let $G=(V,E)$ be the complete graph on $2n$ vertices,
i.e., $G=K_{2n}$. Consider the perfect matching polytope of $G$, given by
\[\mathrm{PM}(G)=\mathrm{conv}\Big(\Big\{\chi_{M}\in\mathbb{R}^{E}:M\subseteq E\,\;\mathrm{ is\,a\,perfect\,matching}\Big\}\Big).\]
We define a linear map $f:P_{\mathrm{config}}(J,M,T)\rightarrow \mathrm{PM}(G)$
as follows: for each edge $e=\{u,v\}\in E$ we define $f_{e}(y):=\sum_{i\in M}y_{i\{u,v\}}$.
For each vertex $y$ of $P_{\mathrm{config}}(J,M,T)$ it holds that
$f(y)\in \mathrm{PM}(G)$ and therefore $f(P_{\mathrm{config}}(J,M,T))\subseteq \mathrm{PM}(G)$.
On the other hand, let $x$ be a vertex of $\mathrm{PM}(G)$,
i.e., $x$ represents a perfect matching $\{\{u_{1},v_{1}\},\{u_{2},v_{2}\},...,\{u_{n},v_{n}\}\}$
in $G$. Then we can construct a vertex $y$ of $P_{\mathrm{config}}(J,M,T)$
such that $f(y)=x$ as follows: assume that the machines are numbered
$\{1,...,n\}$. Then for each machine $i$ we define $y_{i\{u_{i},v_{i}\}}=1$
and $y_{iC}=0$ for each $C\in\C(T)\setminus\{\{u_{i},v_{i}\}\}$.
Then $f(y)=x$. Rothvoss showed that the extension complexity of $\mathrm{PM}(G)$
is $2^{\Omega(n)}$~\cite{Rothvoss2017}. Our construction implies
that $P_{\mathrm{config}}(J,M,T)$ is an extended formulation for
$\mathrm{PM}(G)$. Therefore, the extension complexity of $P_{\mathrm{config}}(J,M,T)$
is also $2^{\Omega(n)}$.
\end{proof}

\section{An extended formulation with small parameterized size}
In this section we consider an abstraction of the configuration integer program that exhibits low extension complexity parameterized by $T$.
Instead of encoding configurations as subsets of jobs, we consider how many jobs in the configuration have a certain processing time $p$.
That is, in this context a configuration is a multiset $C$ of $\{p_j:j\in J\}$, and let $m(p,C)$ be the multiplicity of $p$ in $C$, namely, the number of times that $p$ appears in $C$. 
For every $p\in \{p_j:j\in J\}$, let $n_p$ be the number of jobs in $J$ with processing time equal to $p$.
We denote by $\cmulti(T)$ the set of configurations having total processing time at most $T$. 
The key fact about this encoding is that two or more machines can be scheduled in the same configuration and 
there are fewer solutions overall than for the configuration-LP. 
As before, we propose a formulation where for every combination of machine $i\in M$ and configuration $C\in \cmulti(T)$ we have variable $y_{iC}$ indicating whether $i$ is scheduled according to $C$.
That is,
\begin{alignat}{2}
	 \sum_{C\in \cmulti(T)} y_{iC}                       & =  1   &&\quad\text{for all }i\in M, \text{ } \\
  \sum_{i\in M}\sum_{C\in \cmulti(T)} m(p,C)y_{iC}  & =   n_p &&\quad\text{for all } p \in \{  p_j: j\in J \} \text{, }\\
						    y_{iC}              & \in \{0,1\}   &&\quad\text{for all }i\in M, \text{ for all } C\in \cmulti(T).
\end{alignat}

\begin{theorem}
\label{thm:multi-conf}
The extension complexity of the integer hull for the above formulation is $O(f(T)\cdot \text{poly}(n,m))$ for some function $f$.
\end{theorem}

We define a program deciding for each configuration $C\in \cmulti(T)$ how many machines get configuration $C$ such that we have enough 
slots for all jobs, that is,
\begin{alignat}{2}
	 \sum_{C\in \cmulti(T)} z_{C}                       & =  m   &&,\quad\text{}\; \text{ } \\
  \sum_{C\in \cmulti(T)} m(p,C)z_{C}  & =   n_p &&\quad\text{for all } p \in \{  p_j: j\in J \} \text{, }\\
						    z_C              & \in \mathbb{N}   &&\text{   }\; \text{ for all } C\in \cmulti(T).
\end{alignat}
We denote $P_*(J,M,T)$ the convex hull of the solutions to the integer program above. 
Since the number of processing times ranges in $[1,T]$, the total number of configurations in $\cmulti(T)$ is upper bounded by $T^T$.
\begin{lemma}
\label{lem:extreme-number}
There exists a function $h$ such that the number of vertices of the polytope $P_*(J,M,T)$ is upper bounded by $h(T)$.
\end{lemma}

Observe that the lemma above holds even when number of machines and jobs on the right hand side of the integer program above are not necessarily bounded by a function of $T$, which is in general the case. 
Now we can prove Theorem~\ref{thm:multi-conf} via a Dantzig-Wolfe reformulation with a variable for each of the $h(T)$ vertices of $P_*(J,M,T)$.
Since the size of $P_*(J,M,T)$ is bounded by $h(T)$, we have at most that many variables of this type.
Then, we model the values of the variables $y_{iC}$ as the solution to a suitable transshipment problem.

\begin{proof}[Proof of Lemma~\ref{lem:extreme-number}]
Let $w$
be a vertex of the polytope $P_{*}(J,M,T)$. We argue that there are
at most $h(T)$ possibilities for $w$ for some suitable function
$h$. For a value $g(T)$ to be defined later define a new point $w^{1}$
by setting $w_{C}^{1}:=w_{C}$ if $w_{C}\le g(T)$ and $w_{C}^{1}:=0$
otherwise for each configuration $C\in\cmulti(T)$. Define $w^{2}=w-w^{1}\ge0$,
which is feasible for the following integer linear program 
\begin{alignat*}{2}
\sum_{C\in\cmulti(T)}z_{C} & =m-\sum_{C\in\cmulti(T)}w_{C}^{1}\; , &  & \quad\text{}\;\text{ }\\
\sum_{C\in\cmulti(T)}m(p,C)z_{C} & =n_{p}-m(p,C)w_{C}^{1} &  & \quad\text{for all }p\in\{p_{j}:j\in J\}\text{, }\\
z_{C} & \in\mathbb{N} &  & \text{ }\;\text{ for all }C\in\cmulti(T).
\end{alignat*}
One can easily show that $w^{2}$ is a vertex of the integer hull
of the solutions to the program above, since otherwise $w$ is non-trivial
convex combination of feasible solutions of $P_{*}(J,M,T)$. Let $\cmulti'(T)\subseteq\cmulti(T)$
denote the configurations in the support of $w^{2}$. We identify
each configuration $C\in\cmulti'(T)$ with a vector which is the column
in the matrix of the program above corresponding to $C$, i.e., the
first entry of each such vector is a 1 and the other entries are given
by the values $m(p,C)_{p\in\{p_{j}:j\in J\}}$. If the vectors in
$\cmulti'(T)$ are linearly independent then the above IP has a unique
solution. Hence, there is at most one vertex $w$ of $P_{*}(J,M,T)$
of the form $w=w^{1}+\hat{w}$ where $\hat{w}$ is a solution to the
above IP. In the sequel we argue that it cannot be that the vectors
in $\cmulti'(T)$ are linearly dependent if $g(T)$ is sufficiently
large. Then the claim of the lemma follows since the vector $w^{1}$
satisfies that $w_{C}^{1}\in\{0,...,g(T)\}$ for each $C\in\cmulti(T)$
and the number of such vectors can be bounded by a value $h(T)$. 

If the vectors in $\cmulti'(T)$ are not linearly independent then
there exists a configuration $C\in\cmulti'(T)$ and a set of linearly
independent vectors $C_{1},...,C_{k}\in\cmulti'(T)$ and values $\lambda_{1},...,\lambda_{k}$
such that $C=\sum_{i=1}^{k}\lambda_{i}C_{i}$. Let $\Lambda$ denote
the smallest value $\Lambda'\in\mathbb{N}$ such that $\Lambda'C=\sum_{i=1}^{k}\Lambda'\lambda_{i}C_{i}$
such that $\Lambda'\lambda_{i}\in\mathbb{Z}$ for each $i$. If $g(T)\ge\Lambda$
and $g(T)\ge\Lambda\lambda_{i}$ for each $i$ then we can write $w^{2}$
as the convex combination of two integral vectors and hence $w^{2}$
is not a vertex. To ensure the former we define $g(T)$ to be the
maximum over all values $\Lambda$ and $\Lambda\lambda_{i}$ that
we can obtain in this way, i.e., by selecting one $C\in\cmulti(T)$,
expressing it as a linear combination $C=\sum_{i=1}^{k}\lambda_{i}C_{i}$
of a set of linearly independent configurations $C_{1},...,C_{k}\in\cmulti(T)$,
and finding the smallest value $\Lambda'\in\mathbb{N}$ such that
$\Lambda'\lambda_{i}\in\mathbb{Z}$. Note that the number of values
$\Lambda'$ obtained in this way is finite and hence $g(T)$ is well-defined
(and finite). %
\end{proof}

\begin{proof}[Proof of Theorem~\ref{thm:multi-conf}]
Let $V(J,M,T)$ be the set of vertices in the polytope $P_*(J,M,T)$.
By Lemma~\ref{lem:extreme-number}, there exists a function $h$ such that the size $V(J,M,T)$ is bounded by $h(T)$.
Since we can restrict the optimization problem to its set of vertices, we consider the linear program obtained by lifting the integer program above by the Dantzig-Wolfe reformulation using the set of vertices $V(J,
M,T)$.
For each vertex $v\in V(J,M,T)$, consider a variable $\lambda_v$ indicating whether we pick or not the vertex solution $v$.  
In addition, for each configuration we consider as before a variables $z_C$ indicating how many times the configuration $C$ is used.
Finally, for each combination of machine $i\in M$ and configuration $C\in \cmulti(T)$ we have a variable $y_{iC}$ indicating whether machine $i$ is scheduled according to $C$.
 
The idea behind the extended formulation is to first select a vertex $v\in V(J,M,T)$ by using the Dantzig-Wolfe reformulation in the variables $(\lambda,z)$, which provides for each configuration $C\in \cmulti(T)$ the number of times, $z_C$, that is used. 
Then we formulate a transportation problem between machines and configurations satisfying the offer $z_C$ for each configuration.

More specifically, consider a complete bipartite graph $G$ where we have one vertex $v_i$ for each machine $i\in M$, and we have a vertex $w_C$ for each configuration $C\in \cmulti(T)$.
For each vertex $w_C$ we have an offer $z_C$, and every vertex $v_i$ has a demand of 1.
The variable $y_{iC}$ indicates whether the demand of the machine vertex $v_i$ is satisfied by the configuration vertex $w_C$. 
More specifically, consider 
\begin{alignat}{2}
	 \sum_{v\in V(J,M,T)}    \lambda_{v}                   & =  1,   &&\\
  \sum_{v\in V(J,M,T)}v\lambda_v  & =   z, &&\\
	\sum_{C\in \cmulti(T)} z_{C}                       & =  m,   && \label{eq:offer}\\
  \sum_{C\in \cmulti(T)} m(p,C)z_{C}  & =   n_p &&\quad\text{for all } p \in \{  p_j: j\in J \} \text{, }\label{eq:shape}\\
						    \sum_{i\in M}y_{iC}              & = z_C   &&\text{   }\; \text{ for all } C\in \cmulti(T),\label{eq:trans1}\\
						    \sum_{C\in \cmulti(T)}y_{iC}              & = 1   &&\text{   }\; \text{ for all } i\in M,\label{eq:trans2}\\
						     \lambda_v              & \ge 0   &&\text{   }\; \text{ for all } v\in V(J,M,T),\\	
						    z_C              & \ge 0   &&\text{   }\; \text{ for all } C\in \cmulti(T),\\	
						    y_{iC}              & \ge 0   &&\text{   }\; \text{ for all }i\in M, \text{ for all } C\in \cmulti(T).								    		
\end{alignat}
Observe that constraints (\ref{eq:offer}) and (\ref{eq:trans2}) guarantees that the total demand equals the total offer in the transportation problem over $G$.
It holds that a vertex $(\lambda,z,y)$ of the linear program above is integral. 
If not, suppose that $\lambda$ is fractional, otherwise the integrality of $\lambda$ implies the integrality of $z$ and in turns the integrality of $y$ since the transportation program in the graph $G$ given by constraints (\ref{eq:trans1})-(\ref{eq:trans2}) is integral. 
If $\lambda$ is fractional, $z$ is a non-trivial convex combination of the vertices $\{v\in V(J,M,T):\lambda_v>0\}$, and each of these vertices is feasible for the constraints (\ref{eq:offer})-(\ref{eq:shape}),
which implies that $(\lambda,z)$ is a convex combination of the vectors $\{(e_v,v):\lambda_v>0\}$, where $e_v\in \{0,1\}^{V(J,M,T)}$ is the canonical vector that is 1 for entry $v$ and zero otherwise.
For each $v$ such that $\lambda_v>0$, the constraints (\ref{eq:trans1})-(\ref{eq:trans2}) solve a transportation problem between the machine vertices $\{v_i:i\in M\}$ and the configurations vertices $\{w_C:C\in \cmulti(T)\}$, where the offer for $w_C$ is equal to $z_C$.
The vertices with positive offer are given by $\mathcal{C}_v=\{w_C:v_C>0\}$, and then any solution to this problem is a convex combination of the integral solutions to the transportation problem over $G$.
Since every solution is supported over $M\times \mathcal{C}_v$, and therefore we contradicted the fact that $(\lambda,z,y)$ is a vertex of the polytope. 
\end{proof}

\bibliographystyle{plain}
\bibliography{citations}
\end{document}